\documentclass[12pt,twoside,reqno]{amsart}

\usepackage{amssymb,amsmath,amstext,amsthm,amsfonts,xcolor, amsthm,amscd}

\usepackage{dsfont}

\usepackage{textgreek}

\usepackage[ansinew]{inputenc} 
\usepackage{graphicx}
\usepackage[mathscr]{eucal}

\usepackage{hyperref, enumerate}

\newcommand{\R}{\mathbb{R}}
\newcommand{\C}{\mathbb{C}}
\newcommand{\N}{\mathbb{N}}
\newcommand{\Z}{\mathbb{Z}}

\newcommand{\T}{\mathbb{T}}

\newcommand{\Sym}{{\rm Sym}}
\newcommand{\Mat}{{\rm Mat}}

\newcommand{\diag}{\mbox{diag}}

\newcommand{\DC}{{\rm DC}}

\newcommand{\strip}{\mathscr{A}}
\newcommand{\mscale}{\mathscr{M}}
\newcommand{\paths}{\mathscr{P}_{\alpha' \to \alpha}}
\newcommand{\perms}{\pmb{\sigma}}
\newcommand{\newg}{\hat{g}}
\newcommand{\minor}{\mu_{N, (\alpha, \alpha')}}
\newcommand{\Greene}{G_{N, (\alpha, \alpha')}}
\newcommand{\Greenes}{G_{\scale, (\alpha, \alpha')}}
\newcommand{\bgreene}[1]{G_{#1, (n, n')}}
\newcommand{\meas}{{\rm meas}}

\newcommand{\sign}{\operatorname{sign}}

\newcommand{\sbad}{\mathscr{S}}
\newcommand{\snorm}[1]{\bigl\| #1\bigr\|} 
\newcommand{\sabs}[1]{\left| #1 \right|} 
\newcommand{\abs}[1]{\bigl| #1 \bigr|} 
\newcommand{\babs}[1]{\Bigl| #1 \Bigr|} 

\newcommand{\norm}[1]{\left\Vert#1\right\Vert} 

\newcommand{\normtwo}[1]{
{\left\vert\kern-0.25ex\left\vert\kern-0.25ex\left\vert #1 
    \right\vert\kern-0.25ex\right\vert\kern-0.25ex\right\vert} }

\newcommand{\less}{\lesssim}
\newcommand{\more}{\gtrsim}

\newcommand{\cost}{\mathfrak{c}}

\newcommand{\ep}{\epsilon} 
  
\newcommand{\la}{\lambda}
\newcommand{\ga}{\gamma}
\newcommand{\Ga}{\Gamma}

\newcommand{\om}{\omega}

\newcommand{\vpsi}{\vec{\psi}}

\newcommand{\Symcocycles}{C^{\om}_{r} (\T, \Sym_l(\R))}

\newcommand{\analyticf}[1]{C^{\om}_{r} (\T^{#1}, \R)}
\newcommand{\qpcmat}[1]{C^{\om}_{r} (\T^{#1}, \Mat_m (\R))}

\newcommand{\qpcmatl}[1]{C^{\om}_{r} (\T^{#1}, \Mat_l (\R))}

\newcommand{\scale}{\mathscr{N}}  

\newcommand{\B}{\mathscr{B}}

\newcommand{\dist}{{\rm dist}}

\newcommand\restr[2]{{
  \left.\kern-\nulldelimiterspace 
  #1 
  \vphantom{\big|} 
  \right|_{#2} 
  }}

\newcommand{\Bigo}{\mathcal{O}}

\theoremstyle{plain}
\newtheorem{theorem}{Theorem}
\newtheorem{proposition}{Proposition}

\newtheorem{lemma}[proposition]{Lemma}

\numberwithin{equation}{section}

\theoremstyle{remark}
\newtheorem{remark}{Remark}

\theoremstyle{definition}
\newtheorem{definition}{Definition}

\title[Localization for block Jacobi operators]{Anderson localization for one-frequency quasi-periodic block Jacobi operators}
\date{}

\begin{document}

\author{Silvius Klein}
\email{silviusk@impa.br}
\address{Departamento de Matem\'atica, Pontif\'icia Universidade Cat\'olica do Rio de Janeiro (PUC-Rio), Brazil}

\begin{abstract} 
We consider a one-frequency, quasi-periodic, block Jacobi operator, whose  blocks are generic matrix-valued analytic functions. We establish Anderson localization for this type of operator under the assumption that the coupling constant is large enough but independent of the frequency.
This generalizes a result of J. Bourgain and S. Jitomirskaya on localization for band lattice, quasi-periodic Schr\"o\-din\-ger operators.
\end{abstract}

\maketitle


\section{Introduction and statement}\label{introduction}

An integer lattice  quasi-periodic  Schr\"odinger operator is an operator $H_\la (x)$ on $l^2 (\Z) \ni \psi = \{\psi_n\}_{n \in \Z}$, defined by
\begin{equation*}\label{ s op}
[ H_\la (x) \, \psi ]_n := - (\psi_{n+1} + \psi_{n-1} -2 \psi_n) + \la \, f (x + n \om) \, \psi_n\,,
\end{equation*} 
where  $\la \neq 0$ is a coupling constant, $x \in \T=\R/\Z$ is a phase parameter that introduces some randomness into the system, $f \colon \T \to \R$ is a (bounded) potential function,   and $\om \in \T$ is a fixed irrational frequency.

Note that $H_\la (x)$ is a bounded, self-adjoint operator. Moreover, 
due to the ergodicity of the system, the spectral pro\-per\-ties of the family of operators $\{ H_\la (x) \colon x \in \T \}$ are independent of $x$ almost surely. 

\smallskip

In this paper we study  a more general Schr\"odinger-like operator on a  {\em band} integer lattice (which in some sense may be regarded as an approximation of a higher dimensional lattice).
Before we define such operators, let us introduce some notations and terminology.

\smallskip

All throughout, if $ m \in \N$ and if $M \colon \T \to \Mat_m (\R)$ is any matrix-valued function, we denote by $M^\top (x)$ the transpose of $M (x)$. 

We say that $M (x)$ has no constant eigenvalues if for any $w \in \C$, we have $\det [ M (x) - w \, I ] \not\equiv 0$ as a function of $x$. 

Furthermore, given a frequency $\om \in \T$,  for all $n \in \Z$ we denote by $M_n (x)$ the quasi-periodic matrix-valued function
$$M_n (x) := M (x + n \om) .$$   

All such matrix-valued functions $M$ will be assumed real {\em analytic} (mea\-ning that their entries are real analytic functions).  

Let us then denote by $\qpcmat{}$  the space 
 of  all analytic functions $M \colon \T \to \Mat_m(\R)$ having a holomorphic, continuous up to the boundary extension\footnote{We warn the reader that we identify the torus $\T = \R/\Z$ (an additive group)  with the unit circle $\mathbb{S}^1 \subset \C$ (a multiplicative group) via the map $x + \Z \mapsto e (x) := e^{2 \pi i x}$, but we maintain the additive notation, e.g. we write $x+\omega$ instead of $e (x) e(\omega)$.}
 to $\strip_r := \{ z \in \C \colon 1-r < \sabs{z} < 1 + r \}$, the annulus of  width $2 r$ around the torus $\T$.  
 We endow this space with the uniform norm $\norm{M}_r := \sup_{z\in\strip_r} \norm{M (z)}$. 

\medskip

Let $l \in \N$ be the width of the band lattice, fix an irrational frequency $\om \in \T$ and let $W, R, F \in \qpcmatl{}$. Assume that for all phases $x \in \T$, $R(x)$ and $F(x) \in \Sym_l (\R)$, i.e they are symmetric matrices.

\smallskip

A {\em quasi-periodic block Jacobi operator} is an operator  $H = H_{\la} (x)$ acting on $  l^2 (\Z \times \{1, \ldots, l\},  \, \R) \simeq l^2 (\Z, \R^l)$ by
\begin{equation}\label{J-op}
[ H_{\la} (x)  \, \vpsi ]_n := - (W_{n+1} (x) \, \vpsi_{n+1} + W^{\top}_{n} (x) \, \vpsi_{n-1} + R_{n} (x) \, \vpsi_{n}) + \la \, F_n (x) \, \vpsi_n ,
\end{equation}
where $\vpsi = \{ \vpsi_n \}_{n\in\Z} \in l^2 (\Z, \R^l)$ is any state, and as before, $x \in \T$ is a phase and $\la \neq 0$ is a coupling constant.

This model contains all quasi-periodic, finite range hopping Schr\"{o}\-din\-ger operators on integer or band integer lattices . The hopping term is given by the ``weighted'' Laplacian: 
\begin{equation*}\label{w-Laplace}
[ \Delta_W (x) \, \vpsi]_n :=   W_{n+1} (x) \, \vpsi_{n+1} + W^{\top}_{n} (x) \, \vpsi_{n-1} + R_{n} (x) \, \vpsi_{n} \, ,
\end{equation*}
where the hopping amplitude is encoded by the quasi-periodic matrix-valued functions $W_n (x)$ and $R_n (x)$.

The potential is the quasi-periodic matrix-valued function $\la \, F_n (x)$. 

\medskip

Denote $V (x) := \la F (x) + R (x)$. Since $H = H_{\la} (x)$ acts on $l^2 (\Z, \R^l)$, this operator can be represented in matrix form as an infinite, tri\-diagonal {\em block} matrix, whose building blocks are matrices in $\Mat (l, \R)$: 
$$
H = \left[\begin{array}{cccc}
\ddots & \ddots & 0 & \ddots\\
\\
- W_n^\top & V_n & - W_{n+1} & 0\\
\\
0 & - W_{n+1}^\top & V_{n+1} & - W_{n+2} \\
\\
\ddots & 0 & \ddots & \ddots
\end{array}\right] \, .
$$

\pagebreak

\begin{definition}
We say that an operator satisfies Anderson localization if it has pure point spectrum with exponentially decaying eigenfunctions.
\end{definition}

We may now  formulate the main result of this paper.

\begin{theorem}\label{thm}
Consider the quasi-periodic block Jacobi operator $H_\la (x)$ defined in~\eqref{J-op}, and assume that
\begin{subequations}\label{assumptions}
\begin{align}
\det [ W (x) ]  \not \equiv 0,  \\ 
F (x)   \text{ has no constant eigenvalues.}  
\end{align}
\end{subequations}

There is $\la_0 = \la_0 (W, R, F, l) < \infty $ so that if $\sabs{\la} > \la_0$ and $x \in \T$, then for almost every frequency  $\omega \in \T$, the operator $H_\la (x)$ satisfies An\-derson localization.
 \end{theorem}

  \begin{remark}
 This theorem generalizes a result of J. Bourgain and S. Jitomirskaya~\cite{BJ-band}. The Anderson localization result proven in~\cite{BJ-band}  corres\-ponds to the special case  $W (x) \equiv I_l \in \Mat_l (\R)$,  $$R (x) \equiv R:=  \left[\begin{array}{cccc}
\ddots & \ddots & 0 & \ddots\\
1 & 0 & 1 & 0\\
0 & 1 & 0  &1 \\
\ddots & 0 & \ddots & \ddots
\end{array}\right] \in \Mat_l (\R)$$ and $F (x) \equiv \diag \, \left [ f_1 (x), \ldots f_l (x) \right ]$, where the diagonal entries of $F$ are non-constant functions $f_j \in \analyticf{}$.
 \end{remark}
 
 \begin{remark}
P. Duarte and S. Klein have shown (see Theorem 3.1 in \cite{DK1}) that the sets of matrix valued analytic functions $W \in \qpcmat{}$ and respectively $F \in \Symcocycles$ satisfying the assumptions~\eqref{assumptions} of the theorem  are {\em generic} in a strong sense, namely they are open, dense and prevalent in their respective spaces.

\smallskip

Furthermore, under these same assumptions on the data, and for every frequency $\om$,  it was shown (see Theorem 2.3 in \cite{DK1}) that all nonnegative Lyapunov exponents of the operator $H_\la (x)$ have a lower bound of the type $\log \sabs{\la} - \Bigo (1)$.
 
 \smallskip
 
Recently,  in the same setting and under the same assumptions, P. Duarte and S. Klein have shown (see Corollary 6.1 in~\cite{Coposim}) that if the frequency $\om$ is Diophantine, then the integrated density of states 
of the operator $H_\la (x)$ is weak-H\"older continuous. 
 \end{remark}
 

Given a finite set $\scale \subset \Z$, denote by $H_\scale$ the finite-volume restriction of $H$ to $l^2 (\scale, \R^l) \simeq l^2 (\scale \times \{1, \ldots, l\}, \R)$.
When $\scale = \{1, \ldots, N\}$ for some $N \ge 1$, we use the shorthand notation $H_N$ instead.
Thus 
$$
H_N = \left[\begin{array}{cccc}
V_1 & - W_2 & 0 & \\
\\
- W_2^\top & V_2 & \ddots & \ddots \\
\\
\ddots & \ddots & \ddots & - W_N\\
\\
& 0  & - W_{N}^\top & V_N \\
\end{array}\right]
$$
is an $N \times N$ block matrix-valued function. Each block is an $l \times l$ analytic matrix. Hence $H_N (x)$ may also be regarded as an $N l \times N l$ matrix. 

\smallskip

Given an energy $E \in \R$, the corresponding Dirichlet determinant \\$\det \left [ H_N (x) - E \, I_{N l} \right ]$ and the minors of the matrix $H_N (x) - E \, I_{N l}$ will play a crucial role in our analysis. 

\smallskip

We derive upper bounds on the minors, uniformly in the phase  (see Section~\ref{upper-bound}), and lower bounds for the Dirichlet determinant, on average\- (see Section~\ref{lower-bound}). From these bounds we derive, using Cramer's rule,  estimates on the entries of the Green's function
$G_N (x; E) := \left( H_N (x) - E \, I_{N l} \right)^{-1}$, which in turn lead to the proof of localization (see Section~\ref{localization}).

\smallskip

The main challenge of this paper, and the place where our argument differs most from the one in~\cite{BJ-band}, is deriving the average {\em lower bounds}. To that end, we use the method introduced in P. Duarte and S. Klein~\cite{DK1} for obtaining lower bounds on Lyapunov exponents of quasi-periodic cocycles. 

\smallskip

We note that this method depends upon Hardy's convexity theorem for subharmonic functions, and thus it is a {\em one} variable argument. Therefore, the several variables model (where the potential and weight functions are defined on $\T^d$ with $d>1$, and $\om \in \T^d$ is a translation vector) requires a different set of tools, and it will be considered in a separate project.


\section{An upper bound on minors of Dirichlet matrices}\label{upper-bound}

To ease the presentation, given a block matrix $g$, we use roman letters for the indices of its block-matrix entries, and greek letters for the indices of its scalar entries. More precisely, we write 
$$g = ( g_{\ga, \ga'} )_{1\le \ga, \ga' \le N l} \in \Mat_{N l} (\C) ,$$
which is identified with the block matrix
$$g = ( \underline{g}_{n, n'} )_{1\le n, n' \le N} \in \Mat_N ({\rm R}), \  \text{ where } \, {\rm R} = \Mat_l (\C) .$$

Moreover, given $\ga \in \{1, \ldots, N l \}$, let $n (\ga) \in \{1, \ldots, n \}$ such that $$\ga = l \cdot n (\ga) + r \ \text{ with } \, - l \le r < 0 .$$

Hence given a pair of  indices $(\ga, \ga')$ with $1 \le \ga, \ga' \le N l$, the scalar entry $g_{\ga, \ga'}$ belongs to the $l \times l$-block  \ $\underline{g}_{n(\ga), n(\ga')}$ .


\subsection{Combinatorial preliminaries}

Some of the terminology and the intuition in what follows are related to the lemma of Gessel-Viennot-Lindstr\"om (see Chapter 31 in~\cite{proofs-book} or Chapter 5 in~\cite{enumeration}), although we do not directly use this beautiful result  relating paths in a graph and determinants.

Given a matrix $g = (g_{\ga, \ga'})_{1\le \ga, \ga' \le m} \in \Mat_m (\C)$, it is useful to re\-pre\-sent it as a directed bipartite graph as follows. The vertices of the graph are precisely the rows $R_1, \ldots, R_m$ and the columns $C_1, \ldots, C_m$ of the matrix, while for every pair $(\ga, \ga')$ of indices $\ga, \ga' \in \{1, \ldots, m \}$, there is an edge (or path) $R_\ga \to C_{\ga'}$ with weight (or cost) $g_{\ga, \ga'}$.

A {\em path system} is any collection $\Ga$ of paths
$$R_{\ga_1} \to C_{\ga_1'} \, , \ldots, R_{\ga_k} \to C_{\ga_k'} \, .$$

We will call {\em cost} (or weight) of a path system $\Ga$ the product of the corresponding weights:
$$\cost (\Ga) := g_{\ga_1, \ga_1'} \, \ldots \, g_{\ga_k, \ga_k'} \, .$$

A subset $\mscale \subset \{1, \ldots, m \}$ and a permutation $\sigma \in \perms_\mscale$ determine the following path system
$$\Ga_\sigma := \left\{   R_{\ga} \to C_{\sigma (\ga)} \, \colon \, \ga \in \mscale \right\} \, .$$

Note that $\Ga_\sigma$ is a vertex-disjoint path system with no isolated vertices from $\{ R_\ga \colon \ga \in \mscale \}$ to $\{ C_\ga \colon \ga \in \mscale \}$. Moreover, its cost is
$$\cost (\Ga_\sigma) = \prod_{\ga \in \mscale} \, g_{\ga, \sigma(\ga)} \, .$$

Any subset $\mscale \subset \{1, \ldots, m \}$ is considered ordered with the natural order. We denote by $\neg \mscale$ its complement in $\{ 1, \ldots, m \}$, again, ordered with the natural order. If $\mscale = \{\alpha\}$ is a singleton, we denote $\lnot \{\alpha\}$ by $\lnot \alpha$.

If $\mscale, \mscale' \subset  \{1, \ldots, m \}$  are two non-empty sets of indices, and if $g \in \Mat_m (\C)$, we denote by 
$$g_{\mscale, \mscale'} = (g_{\ga, \ga'})_{\ga \in \mscale, \ga' \in \mscale'}$$
the corresponding reduced matrix.

In particular, if $\alpha, \alpha' \in  \{1, \ldots, m \}$, then $g_{\lnot \alpha, \lnot \alpha'}$ represents the matrix obtained from $g$ by removing row $\alpha$ and column $\alpha'$, so $\det \left[  g_{\lnot \alpha, \lnot \alpha'} \right]$ is the $(\alpha, \alpha')$-minor of $g$.

Let  $\alpha, \alpha' \in  \{1, \ldots, m \}$ with $\alpha \neq \alpha'$. Any tuple $(\ga_1, \ldots , \ga_s)$  with $\ga_1 = \alpha'$, $\ga_s = \alpha$ and $\ga_i \in  \{1, \ldots, m \}$ all {\em distinct} indices determines the path system
$$\Ga = \left\{ (R_\alpha'=) R_{\ga_1} \to C_{\ga_2} \, , R_{\ga_2} \to C_{\ga_3} \, ,  \ldots, R_{\ga_{s-1}} \to C_{\ga_s} (=C_\alpha) \right\} \, .$$

We refer to any such path system simply as ``a path from $R_{\alpha'}$ to $C_\alpha$'' and identify it  with the tuple that determines it, hence $\Ga = (\ga_1, \ldots , \ga_s)$.
Moreover, we denote by $\paths$ the set of all such paths from $R_{\alpha'}$ to $C_\alpha$.

Given $\Ga = (\ga_1, \ldots , \ga_s) \in \paths$, denote by $\sabs{\Ga} = s$ its length and note that its cost is
$$\cost (\Ga) = g_{\ga_1, \ga_2} \, g_{\ga_2, \ga_3} \, \ldots \, g_{\ga_{s-1}, \ga_s} \, .$$

By abuse of notation, $\lnot \Ga$ will denote the complement of the set $\{ \ga_1, \ldots , \ga_s \}$ in $ \{1, \ldots, m \}$.

\smallskip

The following lemma shows how to express a minor of a matrix in terms of the costs of such paths. It was proven by J. Bourgain and S. Jitomirskaya (see Lemma 10 in~\cite{BJ-ac}). We reformulate it here using the terminology introduced above and then present a more detailed proof. 

\begin{lemma}\label{det and paths}
Let $g \in \Mat_m (\C)$ and let $\alpha, \alpha' \in  \{1, \ldots, m \}$ with $\alpha \neq \alpha'$. Then
$$\det \left[  g_{\lnot \alpha, \lnot \alpha'} \right] = (-1)^{\alpha+\alpha'} 
\sum_{\Ga \in \paths} (-1)^{\sabs{\Ga}+1} \, \cost (\Ga) \, \det \left[  g_{\lnot \Ga, \lnot \Ga} \right] \, .$$
\end{lemma}
In other words, the $(\alpha, \alpha')$-minor of $g$ can be expressed as a signed sum over paths from $R_{\alpha'}$ to $C_\alpha$ of the product between the cost of the path and the determinant of the reduced matrix obtained by removing the rows and columns corresponding to vertices in the path.

\begin{proof}
Let $\newg \in \Mat_m (\C)$ be the matrix obtained from $g$ by replacing its $(\alpha, \alpha')$-entry with $1$, and all other entries on row $\alpha$ or column $\alpha'$ with $0$ (while keeping all other entries the same). 

Since $g_{\lnot \alpha, \lnot \alpha'} \in \Mat_{m-1} (\C)$ is the matrix obtained from $g$ by removing row $\alpha$ and column $\alpha'$, we have
$$\det \left[  g_{\lnot \alpha, \lnot \alpha'} \right] = (-1)^{\alpha+\alpha'}  \, \det \left[ \newg \right] .$$

Using the permutations description of the determinant we have:
$$\det \left[ \newg \right] = \sum_{\sigma \in \perms_m} \sign (\sigma) \, \prod_{\ga=1}^m \newg_{\ga, \sigma(\ga)} .$$

In the language of graphs introduced earlier and applied to the matrix $\newg$ and to its corresponding directed bipartite graph, we have
$$\det \left[ \newg \right] = \sum_{\sigma \in \perms_m} \sign (\sigma) \, \cost (\Ga_\sigma) ,$$
or in other words, the determinant of $\newg$ is the signed sum over all vertex-disjoint path systems (from the rows $R_1, \ldots, R_m$ to the columns $C_1, \ldots, C_m$ of $\newg$) of the costs of these path systems. 

It turns out that if $\sigma \in \perms_m$ with $\sigma (\alpha) \neq \alpha'$, then the product
$$\cost (\Ga) = \prod_{\ga=1}^m \newg_{\ga, \, \sigma(\ga)} = 0 ,$$
since it contains the term $\newg_{\alpha, \sigma(\alpha)}$, which is zero (precisely because $\sigma (\alpha) \neq \alpha'$). 

Therefore, the cost of paths $\Ga_\alpha$ corresponding to such permutations is zero, hence
\begin{equation}\label{combinatorics eq1}
\det \left[ \newg \right] = \sum_{\substack{\sigma \in \perms_m \\ \sigma(\alpha)=\alpha'}} \sign (\sigma) \, \cost (\Ga_\sigma) .
\end{equation}
 
Every permutation can be written in a unique way as a product of cycles. Since $\sigma(\alpha)=\alpha'$, there is a unique cycle $\tau$ in this representation of $\sigma$ that contains $\alpha$ and $\alpha'$. 

We may write this cycle as $\tau = (\ga_1, \ldots , \ga_s)$  with $\ga_1 = \alpha'$, $\ga_s = \alpha$ and $\ga_i$ all distinct. We then have the representation
$$\sigma = \tau \cdot \tau^\perp ,$$
where $\tau^\perp$ is a permutation of the set $\{1, \ldots, m \} \setminus \{\ga_1, \ldots , \ga_s\}$ (unless this set is empty, in which case $\tau^\perp$ is the identity permutation).

Note that $\sign (\sigma) = \sign (\tau) \cdot \sign (\tau^\perp)$ and $\sign (\tau) = (-1)^{s+1}$.

Also note that the cycle $\tau$ determines the path (the reader should not mind the abuse of notation) $\Ga =  (\ga_1, \ldots , \ga_s) \in \paths$, and vice-versa.
We can now write
\begin{align*}
\cost (\Ga_\sigma) & = \prod_{\ga=1}^m \newg_{\ga, \, \sigma(\ga)} \\
& = \newg_{\ga_1, \, \sigma(\ga_1)} \cdot \ldots \cdot \newg_{\ga_{s-1}, \, \sigma(\ga_{s-1})} \cdot \newg_{\ga_s, \, \sigma(\ga_s)} \, 
\cdot \prod_{\ga \in \lnot \Ga} \newg_{\ga, \, \sigma(\ga)}  \\
& = g_{\ga_1, \ga_2} \cdot  g_{\ga_2, \ga_3} \cdot \ldots \cdot g_{\ga_{s-1}, \ga_s} \cdot \newg_{\alpha, \alpha'} \cdot
 \prod_{\ga \in \lnot \Ga} g_{\ga, \, \tau^\perp(\ga)}  \\
& = \cost (\Ga) \cdot 1 \cdot   \prod_{\ga \in \lnot \Ga} g_{\ga, \, \tau^\perp(\ga)} = \cost (\Ga) \cdot   \prod_{\ga \in \lnot \Ga} g_{\ga, \, \tau^\perp(\ga)}  \, .
\end{align*}

We rearrange the sum in~\eqref{combinatorics eq1} as
\begin{align*}
\det \left[ \newg \right] & = \sum_{\Ga \in \paths} (-1)^{\sabs{\Ga}+1} \, \cost (\Ga) \, \sum_{\sigma_\star \in \perms_{\lnot \Ga}} \sign (\sigma_\star) \prod_{\ga \in \lnot \Ga} g_{\ga, \, \sigma_\star (\ga)} \\
& = \sum_{\Ga \in \paths} (-1)^{\sabs{\Ga}+1} \, \cost (\Ga) \, \det \left[  g_{\lnot \Ga, \lnot \Ga} \right] \, .
\end{align*}
\end{proof}


\subsection{Statement and derivation of the upper bound}

We fix any coupling constant $\la$ and restrict the set of energies $E$ to a compact interval, say $\abs{E} \less \sabs{\la}$. If  $N \in \N$ and $1 \le \alpha, \alpha' \le N l$, we denote by $\minor (x; E)$  the $(\alpha, \alpha')$-minor of the Dirichlet matrix $H_N (x) - E \, I_{N l}$. That is,
$$\minor (x; E) :=  \det \left[  (H_N (x) - E \, I_{N l})_{\lnot \alpha, \lnot \alpha'} \right] .$$

\begin{proposition}\label{upper bound prop}
 There is $C = C (W, F, R, l) < \infty$ such that for all $N \in \N$, for all indices $1 \le \alpha, \alpha' \le N l$ and for all  $x \in \T$ we have the following uniform upper bound on the corresponding minor:
 \begin{equation*}
 \frac{1}{N l} \, \log \, \abs{ \minor (x; E) } \le \left( 1 - \frac{\abs{ n (\alpha) - n (\alpha') }}{N l} \right) \, \log \sabs{\la} + C .
 \end{equation*}
\end{proposition}

\begin{proof}
Recall that
$$
H_N - E \, I = \left[\begin{array}{cccc}
V_1 - E \, I & - W_2 & 0 & \\
\\
- W_2^\top & V_2 - E \, I & \ddots & \ddots \\
\\
\ddots & \ddots & \ddots & - W_N\\
\\
& 0  & - W_{N}^\top & V_N - E \, I\\
\end{array}\right] 
$$
where $V = \la F + R$. We may assume without loss of generality that $\sabs{\la} \ge 1$.

Let $C = C(W, F, R) < \infty$ be such that  the immediately off-diagonal blocks have the bound $\norm{W}_\infty < C$ and the diagonal blocks have the bound $\norm{V - E \, I}_\infty < C \sabs{\la}$. 

We assume that $\alpha \neq \alpha'$, otherwise the statement  is trivial.

We fix all parameters (i.e. $N, E, x, \alpha, \alpha'$) and let 
$$g := H_N (x) - E \, I \in \Mat_{N l} (\C) , \ \text{ so that }$$
\begin{align}
\minor (x; E) &= \det \left[  g_{\lnot \alpha, \lnot \alpha'} \right]  \notag \\ 
&=
(-1)^{\alpha+\alpha'} 
\sum_{\Ga \in \paths} (-1)^{\sabs{\Ga}+1} \, \cost (\Ga) \, \det \left[  g_{\lnot \Ga, \lnot \Ga} \right] \, , \label{ub eq3}
\end{align}
where the last equality is due to Lemma~\ref{det and paths}.

We estimate the determinant of the matrix $g_{\lnot \Ga, \lnot \Ga}\in \Mat_{N l - \sabs{\Ga}} (\C)$ simply by using Hadamard's inequality:
\begin{equation}\label{ub eq4}
 \abs{ \det \left[  g_{\lnot \Ga, \lnot \Ga} \right]  } \le (C \sabs{\la})^{N l - \sabs{\Ga}}  \, .
 \end{equation}
 
 Combining~\eqref{ub eq3} and \eqref{ub eq4} we have
\begin{equation}\label{ub eq10}
\abs{  \det \left[  g_{\lnot \alpha, \lnot \alpha'} \right]  } \le \sum_{\Ga \in \paths} \abs{ \cost (\Ga) } \, (C \sabs{\la})^{N l - \sabs{\Ga}}  \, .
\end{equation}

In order to estimate the cost $\cost (\Ga)$ of a path $\Ga \in \paths$, we need a more careful analysis.

Let $1 \le \ga, \ga' \le N l$. 

\begin{enumerate}[(i)]
\item If $n (\ga) = n (\ga')$ then $g_{\ga, \ga'}$ belongs to a diagonal block. Thus
$$\abs{ g_{\ga, \ga'} } \le C \sabs{\la} \ \ \text{ (high cost)}.$$

\item If $\abs{n (\ga) - n (\ga')} = 1$ then $g_{\ga, \ga'}$ belongs to an immediately off-diagonal block. Thus
$$\abs{ g_{\ga, \ga'} } \le C  \ \ \text{ (bounded cost)}.$$

\item If  $\abs{n (\ga) - n (\ga')} \ge 2$ then $g_{\ga, \ga'}$ does not belong to a tridiagonal block. Thus
$$\abs{ g_{\ga, \ga'} } = 0  \ \ \text{ (no cost)}.$$
\end{enumerate}

Note also that if $\abs{n (\ga) - n (\ga')} \le 1$ then $\sabs{\ga' - \ga} < 2 l$.

\smallskip

Therefore, in order to estimate the cost of a path $\Ga  \in \paths$ more precisely, we have to take into account the time spent by the path in the diagonal and immediately off-diagonal blocks. 

Indeed, if a path $\Ga = (\ga_1, \ga_2, \ldots, \ga_s) \in \paths$ ventures off the tri\-dia\-go\-nal blocks, that is, if $\sabs{n (\ga_{i+1}) - n (\ga_i)} \ge 2$  for some $1\le i \le s-1$, then 
$g_{\ga_{i}, \ga_{i+1}}  = 0$, so $\cost (\Ga) = 0$. 

\smallskip

We may then restrict ourselves to the collection $\paths^\star$ of paths $\Ga = (\ga_1, \ga_2, \ldots, \ga_s) \in \paths$ with $\sabs{n (\ga_{i+1}) - n (\ga_i)} \le 1$ for all $1\le i \le s-1$. 

\smallskip

Define for such a path
$$b = b (\Ga)  := {\rm card } \, \{ 1 \le i \le s-1  \colon \sabs{n (\ga_{i+1}) - n (\ga_i)} = 1\} $$
to be the number of steps of {\em bounded cost} taken by $\Ga$.

Since $\ga_i \in \{ 1, \ldots, N l \}$ are all distinct, $s \le N l$, so $b (\Ga) < N l$. Also,
$$\abs{ n (\alpha) - n (\alpha') } = \abs{ n (\ga_s) - n (\ga_1) } \le \sum_{i=1}^{s-1} \abs{ n (\ga_{i+1}) - n (\ga_i) } = b (\Ga) .$$

Then for every path $\Ga \in \paths^\star$ we have
$$\abs{ n (\alpha) - n (\alpha') }  \le b (\Ga) \le N l ,$$
so its cost has the bound
\begin{equation}\label{eq5}
\sabs{ \cost (\Ga)} = \prod_{i=1}^{s-1} \sabs{g_{\ga_{i}, \ga_{i+1}}} \le C^{b (\Ga)} \, \left( C \sabs{\la} \right)^{s-b(\Ga)} = C^{\sabs{\Ga}} \, \sabs{\la}^{\sabs{\Ga} - b (\Ga)} \, .
\end{equation}

Furthermore, for any path $\Ga = (\ga_1, \ga_2, \ldots, \ga_s) \in \paths^\star$ we have $\sabs{\ga_{i+1} - \ga_i} < 2 l$ for all indices $1\le i \le s-1$. Since $s \le N l$, 
it follows that there are at most $(4 l)^{N l}$ paths in $\paths^\star$.

Therefore, using~\eqref{ub eq10} and summing over all paths $\Ga \in \paths^\star$ with the same number $b = b (\Ga)$ of steps of bounded cost, we have
\begin{align*}
\abs{  \det \left[  g_{\lnot \alpha, \lnot \alpha'} \right]  } & \le \sum_{b = \sabs{ n (\alpha) - n (\alpha') }}^{N l} \ \sum_{\substack{\Ga \in \paths^\star \\ b (\Ga) = b}}  
\abs{\cost (\Ga)} \,  \left( C \sabs{\la} \right)^{N l - \sabs{\Ga}} \\
& \le \sum_{b = \sabs{ n (\alpha) - n (\alpha') }}^{N l} \ C^{\sabs{\Ga}} \, \sabs{\la}^{\sabs{\Ga} - b }  \,  \left( C \sabs{\la} \right)^{N l - \sabs{\Ga}} \, (4 l)^{N l} \\
& = (4 l C)^{N l} \, \sabs{\la}^{N l} \,  \sum_{b = \sabs{ n (\alpha) - n (\alpha') }}^{N l} \sabs{\la}^{-b} \\
& \less (4 l C)^{N l} \, \sabs{\la}^{N l} \, \sabs{\la}^{-  \sabs{ n (\alpha) - n (\alpha') }} \, ,
\end{align*}
which completes the proof, after taking logarithms.  
\end{proof}

\section{A lower bound on Dirichlet determinants}\label{lower-bound}

\begin{proposition}\label{lower bound prop}
There are finite constants $\la_0 = \la_0 (W, F, R)$ and $C = C (W, F, R)$ such that for all $ \om \in \R$, $N \in \N$ and $\la$ with $\sabs{\la} \ge \la_0$ we have
\begin{equation}\label{lower bound eq}
\int_\T \frac{1}{N l} \, \log \, \abs{ \det \left [ H_N (x) - E \, I_{N l} \right ] }  \, d x \ge \log \sabs{\la} - C 
\end{equation}
for all energies $E$ in a compact interval (say for $\sabs{E} \less \sabs{\la}$).
\end{proposition}

\begin{proof}
The matrix-valued functions $W$ and $V = \la F + R$ are holomorphic on $\strip_r$ and continuous up to the boundary; in particular, for some constant $C < \infty$, $\norm{W}_r \le C$ and  $\norm{V}_r \le C \sabs{\la}$. It follows that the block matrix-valued function $H_N - E\,I$ is also holomorphic on $\strip_r$ and considering only energies $E$ with $\sabs{E} \le C \sabs{\la}$, $\norm{H_N - E\,I}_r \le 2 C \sabs{\la}$.

Let
$$u (z) := \frac{1}{N l} \, \log \, \abs{ \det \left [ H_N (z) - E \, I \right ] } .$$

Then $u (z)$ is {\em subharmonic} on $\strip_r$ and using Hadamard's inequality to estimate the determinant, it satisfies the upper bound
\begin{equation}\label{upper bound 1}
u (z) \le \log \sabs{\la} + \Bigo(1) \quad \text{ for all } z \in \strip_r .
\end{equation} 

We will derive a lower bound for $u (z)$ on a circle $\sabs{z} = 1 + y_0$, where $0 < y_0 \ll r$.

Write 
\begin{equation}\label{eq 0}
H_N (z) - E \, I = \la D_N (z) + B_N (z) ,
\end{equation}
where
$$D_N (z) = \diag \left [ F_1 (z) - \frac{E}{\la} \, I, \ldots, F_l (z) - \frac{E}{\la} \right ] $$
is a block diagonal matrix and 
$$B_N  = \left[\begin{array}{cccc}
R_1 & - W_2 &  \\
- W_2^\top & R_2 & \ddots  \\
 & \ddots & \ddots & - W_N\\
 &  & - W_{N}^\top & R_N \\
\end{array}\right]
$$
is the remaining block tridiagonal matrix.

Note that $$\norm{B_N}_r \le C \quad \text{and } \ \norm{D_N}_r \le C .$$ 

We assume that $F$ has no constant eigenvalues: for any $t \in \R$ we have $\det \left [ F (x) - t\, I \right] \not\equiv 0$ (as a function of $x \in \T$). So given any $0 < \delta \ll r$, there is $y_0 \sim \delta$ and there is $\ep_0 = \ep_0 (\delta, r, F) > 0$ such that 
\begin{equation}\label{det 1}
\sabs{ \det \left [ F (z) - t\, I \right ] } \ge \ep_0^l
\end{equation}
for all $z$ with $\sabs{z} = 1 + y_0$ and for all $t$ with $\sabs{t} \le C$.

This follows from Corollary 4.6 in~\cite{DK1}. We present here an alternative, self contained argument. 

Let $f_t (z) := \det \left [ F (z) - t \, I \right]$. Since $f_t (z)$ is holomorphic and non-identically zero on $\strip_r$, it has a finite number of zeros in any compact set, in particular in the annulus 
$\strip:= \left \{ z \colon 1+\frac{\delta}{2} \le \sabs{z} \le 1 + 2 \delta \right\}$.

Therefore, there are circles $\sabs{z} = 1 + y$ with $\frac{\delta}{2} \le y \le 2 \delta$ on which $f_t (z)$ has no zeros, so
$$\varepsilon (t) :=  \sup_{\frac{\delta}{2} \le y \le \delta} \ \inf_{z \colon \sabs{z} = 1 + y} \  \sabs{ f_t (z) } > 0 .$$

We show below that the map $\R \ni t \mapsto \varepsilon (t)$ is lower semicontinuous. This will imply that it has a positive lower bound when $t$ is restricted to a compact interval, thus establishing \eqref{det 1}.

To prove the lower semicontinuity of $\varepsilon (t)$, fix $t$ and assume that $$\varepsilon (t) > \ep_0^l .$$ 

(Since $l$ is fixed, the power $l$ in $\ep_0^l$ is just for aesthetic reasons related to the calculations that follow). Then there is $y \in \left[ \frac{\delta}{2}, 2 \delta \right]$ such that on the circle $\sabs{z} = 1+y$ we have
$$\sabs{ \det \left [ F (z) - t\, I \right] } = \sabs{ f_{t_0} (z) } \ge \ep_0^l .$$

Hence on $\sabs{z} = 1+y$, the matrix $F (z) - t \, I$ is invertible  and by Cramer's formula,
\begin{equation}\label{det 30}
\norm{ \left(F (z) - t\, I\right)^{-1}  } \le \frac{\norm{{\rm adj} \left[ F (z) - t \, I \right]}}{ \sabs{ \det \left [ F (z) - t\, I \right ] } } \le \frac{C^l}{\ep_0^l} = \left(C \, \ep_0^{-1}\right)^l .
\end{equation}

Furthermore, for any $t'$ we have
\begin{align}
f_{t'} (z) & = \det \left [ F (z) - t' \, I \right] = \det \left [ F (z) - t I + (t - t')  \, I \right]  \notag \\
& = \det \left [ F (z) - t \, I \right] \, \det \left[ I +  (t-t') \, \left( F (z) - t I \right)^{-1} \right] \label{det 40}
\end{align}

To estimate from below the second determinant on the right hand side above, we use the following simple fact. 

If $g \in \Mat_m (\R)$ is an invertible matrix and if we denote by $s_1(g) \ge \ldots \ge s_m (g) > 0$ its singular values, then $\norm{g^{-1}} = s_m(g)^{-1}$, so
$$\sabs{ \det \left [ g \right ] } = s_1(g) \ldots s_m (g) \ge \norm{g^{-1}}^{-m} .$$

Now if $g \in \Mat_m (\R)$ and $\norm{g} < 1$, then $I + g$ is invertible and $\norm{(I+g)^{-1}} \le (1-\norm{g})^{-1}$. By the previous considerations applied to $I+g$ we conclude
\begin{equation}\label{det 20}
\sabs{ \det \left [ I + g \right ] } \ge (1-\norm{g})^{m} .
\end{equation}

Applying \eqref{det 20} to $g = (t-t') \, \left( F (z) - t I \right)^{-1}$ and using \eqref{det 30} (which also ensures that $\norm{g} < 1$ if $\sabs{t-t'} \ll 1$), we have:
\begin{align*}
\sabs{ \det \left[ I +  (t-t') \, \left( F (z) - t I \right)^{-1} \right] } & \ge \left(1 - \sabs{t-t'} \, \left(C \, \ep_0^{-1}\right)^l \right)^l \\
& = 1- \Bigo(\sabs{t-t'}) .
\end{align*}

Combined with \eqref{det 40}, this  implies that on $\sabs{z} = 1+y$ we have
$$\sabs{f_{t'} (z)} \ge \ep_0^l \, \left( 1- \Bigo(\sabs{t-t'}) \right) .$$

Hence 
$$\varepsilon (t') \ge  \inf_{z \colon \sabs{z} = 1 + y}  \  \sabs{ f_{t'}  (z) } \,  \ge \, \ep_0^l \, \left( 1- \Bigo(\sabs{t-t'}) \right) ,$$
which proves the lower semicontinuity of the function $\varepsilon$.

\smallskip

Thus the lower bound \eqref{det 1} holds, and it implies that on the circle $\sabs{z} = 1 + y_0$  the matrix $F (z) - t \, I$ is invertible and as in \eqref{det 30} we have 
$$ \norm{ \left(F (z) - t\, I\right)^{-1}  } \le \left(C \, \ep_0^{-1}\right)^l .$$

Since the complexified dynamics leaves the circle $\sabs{z} = 1 + y_0$  invariant, all the blocks 
$F_j (z) - t \, I = F(z+j\om)-t\,I$ are invertible and their inverses have the same bound as above there.

Therefore, on the circle $\sabs{z} = 1 + y_0$, for all $E$ with $\sabs{E} \le C \sabs{\la}$, the block matrix
$D_N (z) = \diag \left [ F_1 (z) - \frac{E}{\la} \, I, \ldots, F_l (z) - \frac{E}{\la} \, I \right ]$ is invertible as well, and since
$$D_N^{-1} (z) = \diag \left [ \left(F_1 (z) - \frac{E}{\la}\, I\right)^{-1}, \ldots, \left(F_l (z) - \frac{E}{\la} \, I \right)^{-1} \right ] ,$$
we conclude that
\begin{equation}\label{det 5}
\norm{D_N^{-1} (z)} \le \left(C \, \ep_0^{-1}\right)^l .
\end{equation}

Furthermore, using \eqref{det 1} again,
\begin{equation}\label{det 2}
\sabs{ \det \left [ D_N (z) \right ] } = \prod_{j=1}^N \sabs{ \det \left [ F_j (z) - \frac{E}{\la} \, I \right ] } \ge \ep_0^{N l} .
\end{equation}

Going back to \eqref{eq 0}, we can now write for $\sabs{z} = 1 + y_0$
$$H_N (z) - E \, I  = \la D_N (z) \, \left [ I + \la^{-1} \, D_N^{-1} (z) \, B_N (z) \right ] ,$$
so
\begin{equation}\label{det 3}
\sabs{ \det \left [ H_N (z) - E \, I \right ] } = \sabs{\la}^{N l} \,  \sabs{ \det \left [ D_N (z) \right ] }  \, \sabs{\det\left[I+\la^{-1}  D_N^{-1} (z)  B_N (z)\right]} .
\end{equation}

As before, we estimate from below the second determinant on the right hand side of \eqref{det 3} by applying \eqref{det 20}
 to $g := \la^{-1}  D_N^{-1} (z)  B_N (z) \in \Mat_{N l} (\R)$.
Using \eqref{det 5}, on the circle $\sabs{z} = 1 + y_0$ we have
$$\norm{ \la^{-1}  D_N^{-1} (z)  B_N (z) } \le \frac{\left(C \, \ep_0^{-1}\right)^l \, C}{\sabs{\la}} \le \frac{1}{2}$$
provided $\sabs{\la}$ is large enough.
Then
\begin{equation}\label{det 6}
 \sabs{\det\left[I+\la^{-1}  D_N^{-1} (z)  B_N (z)\right] } \ge 2^{- N l} .
\end{equation}

Combining \eqref{det 3}, \eqref{det 2} and \eqref{det 6} we conclude that on $\sabs{z} = 1 + y_0$ we have
\begin{equation}\label{lower bound sh}
u (z) = \frac{1}{N l} \, \log \, \abs{ \det \left [ H_N (z) - E \, I \right ] } \ge \log \sabs{\la} + \log \frac{\ep_0}{2} .
 \end{equation}
 
We gather up below the estimates \eqref{upper bound 1} and \eqref{lower bound sh} on the subharmonic function $u (z)$, defined on the annulus $\strip_r$.
\begin{subequations}
\begin{align*}
u (z) & \le \log \sabs{\la} + \Bigo(1) & \kern-2em \text{on }  & \sabs{z} = 1 + r \\
u (z) & \ge \log \sabs{\la} - \Bigo(1) & \kern-2em \text{on } & \sabs{z} = 1+ y_0 .
\end{align*}
\end{subequations}

By Hardy's convexity theorem (see for instance \cite{Duren}), the mean of a subharmonic function on an annulus is radially $\log$-convex. More precisely, the function
$$(1-r, 1+r) \ni s \mapsto \int_{\sabs{z}=s} u (z) \frac{d z}{2 \pi}$$
is convex as a function of $\log s$.  

Since $1 < 1+y_0 < 1+r$, letting $\alpha:= \frac{\log (1+y_0)}{\log (1+r)} \in (0, 1)$ so that
$$\log (1+y_0) = (1-\alpha) \log (1) + \alpha \log (1+r) ,$$
we then have
$$\int_{\sabs{z}=1+y_0} u (z) \frac{d z}{2 \pi} \le (1-\alpha) \int_{\sabs{z}=1} u (z) \frac{d z}{2 \pi} + \alpha \int_{\sabs{z}=1+r} u (z) \frac{d z}{2 \pi} .$$

Combined with the lower bound on $\sabs{z}=1+y_0$ and the upper bound on $\sabs{z} = 1 + r$, this implies a lower bound for the mean on $\sabs{z}=1$:
\begin{align*}
(1-\alpha) \int_{\T} u (x) \, d x & = (1-\alpha) \int_{\sabs{z}=1} u (z) \frac{d z}{2 \pi}   \\
&\kern-4.5em \ge \log \sabs{\la} - \Bigo(1) - \alpha \left( \log \sabs{\la} + \Bigo(1) \right) 
= (1-\alpha) \log \sabs{\la} - \Bigo(1),
\end{align*}
hence
$$\int_{\T} u (x) \, d x \ge \log \sabs{\la} - \Bigo(1) ,$$ which concludes the proof of this proposition.
\end{proof}

\section{Green's functions estimates and the proof of localization}\label{localization}

Let $\scale = [a, b] = \{ a, a+1, \ldots , b-1, b \} \subset \Z$  be an interval of integers and let $\sabs{\scale} = b-a+1$ be its length. Consider the corresponding {\em Green's function}
$$G_\scale (x) = G_\scale (x; E) := \left( H_\scale (x) - E \, I \right)^{-1} , $$
defined whenever $H_\scale (x) - E \, I$ is invertible. 

We may regard $G_\scale (x; E)$ as a $\sabs{\scale} \times \sabs{\scale}$ block matrix, whose blocks are $l \times l$ matrices over $\R$. In this case we denote its entries by 
$\bgreene{\scale} (x; E)$, where $n, n' \in \scale$ (as before, we use roman letters for the indices of the block entries of such a matrix).

The Green's function may also be regarded as a $\sabs{\scale} l  \times \sabs{\scale} l$ matrix with real entries, which we denote  by  $G_{\scale, (\alpha, \alpha')}  (x; E)$, 
where $(a-1) l +1 \le \alpha, \alpha'  \le b l$.

Furthermore, given a function $\vpsi \colon \Z \to \R^l$, its finite volume restriction is
$$\vpsi_\scale := 
\begin{bmatrix}
\vpsi_a \\
\vdots\\
\vpsi_b
\end{bmatrix}  \in \left(\R^l\right)^{\sabs{\scale}} \simeq \R^{ \sabs{\scale} \, l} \, .
$$

When the integers interval $\scale = [1, N]$ for some $N \ge 1$, we use the shorthand notations $G_N (x; E)$ and $\vpsi_N$ for $G_\scale (x; E)$ and $\vpsi_\scale$ respectively. 

Note that for any $j \in \Z$ we have
$G_\scale (x + j \om; E) = G_{\scale + j} (x; E)$. Hence for any interval $\scale = [a, b] \subset \Z$, 
$$G_\scale (x; E) = G_{\sabs{\scale}} (x + (a-1) \om; E) \, .$$

Recall also the function considered in the previous section
\begin{equation*}
u_N (x) = u_N (x; E) :=  \frac{1}{N l} \log \abs{ \det \left[ H_N (x) - E \, I \right]} \, .
\end{equation*}

\subsection{Green's functions estimates}

\medskip

A crucial ingredient in the proof of Anderson localization for the operator $H_\la (x)$ is the exponential decay of the off-diagonal entries of the Green's function $G_N (x)$. 
These estimates are obtained using a certain concentration of measure inequality for the function $u_N (x) $ as well as 
 the upper and lower bounds in Proposition~\ref{upper bound prop} and Proposition~\ref{lower bound prop}.
 
The concentration of measure bounds  require an arithmetic assumption on the frequency $\om$.  
If $t > 0$ let $\DC_t$ be the set of frequencies $\omega\in\T$ satisfying the following {\em Diophantine condition}: 
\begin{equation}\label{DC}
\norm{k \, \om} \ge \frac{t}{\sabs{k}^{2}}
\quad \text { for all }\; k\in \Z\setminus\{0\}\;,
\end{equation}
where  for any real number $x$ we write $\norm{x}:=\min_{k\in\Z} \abs{x-k}$.

Note that $\displaystyle \cup_{t > 0} \DC_t$ is a set of full measure.

\smallskip

Let us formulate the Green's functions estimates needed in the proof of localization.

\begin{definition}\label{def good Green}
Let $\scale = [a, b] \subset \Z$ be an interval and let $(x, \om, E) \in \T \times \T \times \R$. We say that $G_\scale (x; E)$ is a {\em good} Green's function if for all $n, n' \in \scale$
\begin{equation}\label{good Green}
\norm{ \bgreene{\scale} (x; E)   } \le e^{- (\sabs{n -n'} - \frac{\sabs{\scale}}{50} ) \, \log \sabs{\la}} \ .
\end{equation}

Note that $\bgreene{\scale} (x; E) $ is an $l \times l$ block, and $\norm{ \bgreene{\scale} (x; E) }$ refers to any of the equivalent norms on $\Mat_l (\R)$. Clearly~\eqref{good Green} is equivalent to 
\begin{equation}\label{greens estimate eq}
\abs{ \Greenes (x ; E) } \less e^{- (\sabs{n (\alpha)-n(\alpha')} - \frac{\sabs{\scale}}{50} ) \, \log \sabs{\la}} \, 
\end{equation}
for all $(a-1) l + 1 \le \alpha, \alpha' \le b l$.
\end{definition}

We now describe the setting under which the concentration of measure and consequently the Green's functions estimates will hold.

Fix $t>0$, $\om \in \DC_t$ and $\delta > 0$ small enough, say $\delta = \frac{1}{100 l}$. 

Let $n_0$ be a large enough integer depending on $t$ and on $l$ and let $n_1 \gg n_0$.

Let $\la_0$ be a large enough, finite constant depending on the data (e.g. on $W, F, R, l$); $\la_0$ is essentially the constant from Proposition~\ref{lower bound prop}, but we might slightly increase it so that, for instance, $\log \la_0$ is much larger than the other constants (denoted by $C$) from Propositions~\ref{upper bound prop} and~\ref{lower bound prop}.

Fix any coupling constant $\la$ with $\sabs{\la} \ge \la_0$ and drop it from notations.

The energy $E$ will be such that $\sabs{E} \less \sabs{\la}$.

With this setup, we have the following result.

\begin{proposition}\label{Green's functions estimates}
There are absolute constants $a > 0$ and $p \in \N$ such that if  $M \ge n_0$ and $N \ge n_1$ then the following hold.
\begin{enumerate}[(i)]
\item  \label{Green item 1} Let
\begin{equation}\label{qBET}
 \B_{N}^M = \B_{N}^M (\om, E) :=  \left\{  x \in \T \colon \frac{1}{M} \sum_{j=0}^{M-1} u_N (x+j \om)  \le (1-\delta) \log \sabs{\la}   \right\}    .
\end{equation}

Then $$\meas \left( \B_{N}^M (\om, E) \right) < e^{-M^a} .$$

\item \label{Green item 2}  For every  $x \notin \B_{N}^M (\om, E)$ there is $0 \le j < M$ such that $G_N (x + j \om ; E)$ is a good Green's function.

\item \label{Green item 3}  For every $x \in \T$ there are integers $0 \le n < N^p$ such that $G_N (x + n \om ; E)$ is a good Green's function. In fact,
$$\# \left\{ 0 \le n < N^p \colon G_N (x + n \om ; E)  \text{ is {\em not} a good Green's function}  \right\} \ll N^p  .$$

\end{enumerate}
\end{proposition}

\begin{proof}
(i) Recall from the proof of Proposition~\ref{lower bound prop} that the function $u (x) = u_N (x)$ has a {\em subharmonic} extension 
$$u (z) = u_N (z)  =  \frac{1}{N l} \log \abs{ \det \left[ H_N (z) - E \, I \right]} \, .$$

From \eqref{upper bound 1}, this function has the upper bound
$$u (z) \le \log \sabs{\la} + \Bigo(1) \quad \text{ for all } z \in \strip_r .$$

While it is not bounded from below, we showed in Proposition~\ref{lower bound prop} that on average, it has the bound
$$\int_\T u (x) d x \ge \log \sabs{\la} - \Bigo (1) .$$

Therefore, there is $x_0 \in \T$ such that
$$u (x_0)  \ge \log \sabs{\la} - \Bigo (1) .$$

This is all we need for the quantitative version of Birkhoff's ergodic theorem (showing that the averages of $u(x)$ along the orbits of the torus translation by a Diophantine frequency concentrate near the mean of $u$) to be applicable. See for instance Subsection 6.2.3 in \cite{DK-book} (specifically, Lemma 6.6) for more details. Therefore, say by Theorem 6.5 in \cite{DK-book}, we have that if $\om \in \DC_t$ and $M \ge t^{-2}$, then for some absolute constant $a > 0$ we have
$$\meas  \left\{  x \in \T \colon \babs{ \frac{1}{M} \sum_{j=0}^{M-1} u (x+j \om) - \int_\T u (x) d x }  >  S \, M^{-a}   \right\}   < e^{-M^a} ,$$
where $S$ is of the order of the above bounds on $u$ , hence $S = \Bigo (\log \sabs{\la})$.

 If $x$ does not belong to the set above, then using again the lower bound on the mean of $u$ derived in Proposition~\ref{lower bound prop}, we have 
\begin{align*}
\frac{1}{M} \sum_{j=0}^{M-1} u (x+j \om)  & \ge \int_\T u (x) d x   -  S \, M^{-a} \\
&  \ge \log \sabs{\la} - \Bigo (1) -  C \, \log \sabs{\la}  \, M^{-a} > (1-\delta) \log \sabs{\la} ,
\end{align*}
provided $M > l^{1/a}$ and $\sabs{\la}$ is large enough.

\smallskip

(ii) By Cramer's rule, for all indices $1 \le \alpha, \alpha' \le N l$ we have
\begin{align*}
\Greene (x) = \left( (H_N (x) - E \, I)^{-1} \right)_{(\alpha, \alpha')} =  \frac{ \minor (x) }{  \det \left[ H_N (x) - E \, I \right] } \, .
\end{align*}

Taking logarithms, 
\begin{align*}
\frac{1}{N} \log \abs{ \Greene (x) } & =  \frac{1}{N} \, \log \, \abs{ \minor (x) } -  \frac{1}{N } \log \abs{ \det \left[ H_N (x) - E \, I \right]} \\
& =  \frac{1}{N}  \log \, \abs{ \minor (x) } - l \, u_N (x) .
\end{align*}

By Proposition~\ref{upper bound prop}, for all $x \in \T$ we have the upper bound
\begin{equation*}
 \frac{1}{N} \, \log \, \abs{ \minor (x) } \le  l \, \log \sabs{\la} - \frac{\abs{ n (\alpha) - n (\alpha') }}{N }  \, \log \sabs{\la} + C l.
 \end{equation*}

If  $x  \notin \B_N^M$, then by item~\eqref{Green item 1} we  have
$$\frac{1}{M} \sum_{j=0}^{M-1} u_N (x+ j \om)  >  (1-\delta) \log \sabs{\la}  .$$

Then for some $0 \le j < M$ we must have
$$u_N (x + j \om) > (1-\delta) \log \sabs{\la}  .$$

We conclude that 
\begin{align*}
\frac{1}{N} \log \abs{ \Greene (x + j \om) } & \le  l \log \sabs{\la} -  \frac{\abs{ n (\alpha) - n (\alpha') }}{N }  \, \log \sabs{\la}  + C l \\
& -  l \log \sabs{\la} + \delta  l \log \sabs{\la} \\
& < - \left(  \frac{\abs{ n (\alpha) - n (\alpha') }}{N } - \frac{1}{50}   \right) \log \sabs{\la} \, .
\end{align*}
This implies~\eqref{greens estimate eq} for $\scale = [1, N]$, so $G_N (x + j \om ; E)$ is a good Green's function.

\medskip

(iii) In order to prove this last statement, it is enough to show that there is $p < \infty$ such that for every $x \in \T$, there is $0 \le k < N^p$ with $x + k \om \notin \B_N^M$, where $M$ is an integer with $n_0 \le M \ll  N$.

By the pointwise ergodic theorem, any typical orbit $\{ x + k \om \colon k \ge 0 \}$ will of course visit any set of positive measure. The point is to obtain a precise, quantitative version of this statement. To that extent, the fact that $\om$ is Diophantine and the algebraic structure of the set $\B_N^M$ are crucial.

First recall that $\meas (\B_N^M) < e^{-M^a}$. Furthermore,  the ine\-quality~\eqref{qBET} defining the set $\B_N^M$ may be rewritten as 
\begin{equation}\label{salg 1}
\prod_{j=0}^{M-1} \det \left[ H_N (x+j \om) - E \, I \right] \le \sabs{\la}^{(1-\delta) N M l} \, .
\end{equation}

$H_N (x) - E \, I$ may be regarded as an $N l \times N l$ matrix-valued analytic function. 
Let $f (x)$ be any of its entries. Expand $f (x)$ into its Fourier series $$f (x) = \sum_{m\in\Z} \hat{f}_m e^{2 \pi i \, m \, x}$$  and consider the truncation $$f_1 (x) := \sum_{\sabs{m} \le N^2 } \hat{f}_m e^{2 \pi i \, m \, x} ,$$
so that $\norm{f-f_1}_\infty \less e^{-N^2 } $.

Therefore, substituting each entry of $H_N (x) - E \, I$ by its truncation will produce negligible errors (i.e. super-exponentially small) in ~\eqref{greens estimate eq}.

 Hence we may assume that~\eqref{salg 1} is in fact a {\em trigonometric polynomial} inequality of degree $\le N^2 \cdot N l \cdot M = l N^3 M$. Further truncating  the power series for $\cos$ and $\sin$, we conclude that the set $\B_N^M$ described by~\eqref{salg 1} may be regarded as a {\em semi-algebraic set} of degree at most $\Bigo (N^4 M)$. 

Put $M=N^{1/2}$. We have a set $\B=\B_N^M$  which is contained in a semi-algebraic set of degree at most $\Bigo (N^5)$, whose measure is comparable to that of $\B$, hence it is sub-exponentially (in $N$) small. 

Then if $N_1 := N^p$, where $p$ is a large enough absolute constant, by Coro\-llary 9.7 in J. Bourgain's monograph~\cite{Bourgain-book}, for every $x \in \T$,
$$\#  \{ k = 0, \ldots, N_1 \, \colon \, x + k \om  \in \B \} \ll N_1  .$$

Thus there are (plenty of) integers $0 \le k \le N^p$ for which $x + k \om \notin \B$, which completes the proof.
\end{proof}

\begin{remark}\label{good intervals}
In the proof of localization we will apply item~(\ref{Green item 2}) in Proposition~\ref{Green's functions estimates} with $M$ being a small fraction of $N$, say $M = \frac{1}{100} N$. Then the statement may be reformulated as follows. 

There is an absolute constant $a>0$ and for every large enough integer $N$ there is a set of phases $\B_N = \B_N (\om, E)$ with $\meas \left(\B_N (\om, E) \right) < e^{- N^a}$ such that if $x \notin \B_N (\om, E)$ then $G_\scale (x ; E)$ is a good Green's function for some
$$\scale \in \left\{ [1, N] +j  \, \colon \, 0 \le j <  \frac{1}{100} N \right\} \, .$$

That is, $G_\scale (x ; E)$ is a good Green's function where  $\scale$ is the interval $[1, N]$ or a relatively small displacement thereof. 

With this reformulation, we have the analogue of Proposition 7.19 in J. Bourgain's monograph~\cite{Bourgain-book}. That proposition forms the basis for the proof of localization for quasi-periodic Schr\"{o}dinger operators  described in Chapter 10 of the monograph.

Finally, it is clear that similar statements to those in items~(\ref{Green item 2}) and~(\ref{Green item 3}) of Proposition~\ref{Green's functions estimates} also hold when replacing $G_{[1, N]}$ by $G_{[- N, N]}$.
\end{remark}


\subsection{The sketch of the proof of localization}

\medskip

The proof of localization consists of a parameter elimination argument that uses the Green's functions estimates in Proposition~\ref{Green's functions estimates}  and semi-algebraic sets considerations. The argument is analogous to the one sketched in~\cite{BJ-band}, which in turn references the method introduced by J. Bourgain and M. Goldstein in~\cite{B-G-first} to establish non-perturbative Anderson localization for quasi-periodic Schr\"odinger operators on the integer lattice. We present below a sketch of the  argument and indicate the point where it differs from~\cite{B-G-first}. 

Bourgain-Goldstein's method uses the Sch'nol-Simon theorem\footnote{This result holds
 in our more general setting, for instance as a consequence of the generalization of Sch'nol's theorem obtained in~\cite{RuiHan}. See also 
Theorem 4.1 in~\cite{JM-survey}.}: to esta\-blish localization, it is enough to show that every extended state (i.e. gene\-ralized eigenvector) is exponentially decaying.

More precisely, if for any $\vpsi = (\vpsi_n)_{n\in\Z}\subset\R^l$ and $E \in\R$ such that
\begin{equation}\label{generalized eigenvector}
H  (x) \, \vpsi = E \, \vpsi, \  \snorm{\vpsi_0}_2 =1 \  \text{and}  \ \snorm{\vpsi_n}_2 \less (1+\sabs{n})^2 \ \text{ for all } n \in \Z ,
\end{equation}
one shows that necessarily for some $c > 0$
\begin{equation}\label{exp decay of gev}
\snorm{\vpsi_n}_2 \le e^{-c \, \sabs{n}}  \text{ as } \sabs{n} \to \infty ,
\end{equation}
then $H (x)$ satisfies Anderson localization. 

 \medskip
 
Consider $x \in \T$, $\vpsi = (\vpsi_n)_{n\in\Z}$ and $E\in \R$ such that \eqref{generalized eigenvector} holds. 

Clearly $H  (x) \, \vpsi = E \, \vpsi$ is equivalent to
$$- W_{n+1} (x) \,  \vpsi_{n+1} - W_n^\top (x)  \, \vpsi_{n-1} + (V_n (x) - E \, I) \, \vpsi_n = 0 \quad \text{ for all } \ n \in \Z , $$
which implies that for every interval $\scale = [a, b] \subset \Z$,
\begin{equation*}\label{eq1 AL}
(H_\scale (x) - E \, I) \, \vpsi_\scale = 
\begin{bmatrix}
W_a^\top (x) \, \vpsi_{a-1} \\
\vec{0}\\
\vdots\\
\vec{0}\\
W_{b+1} (x) \, \vpsi_{b+1}
\end{bmatrix} 
=
\begin{bmatrix}
W_a^\top (x) \, \vpsi_{a-1} \\
\vec{0}\\
\vdots\\
\vec{0}\\
\vec{0}
\end{bmatrix} 
+
\begin{bmatrix}
\vec{0}\\
\vec{0} \\
\vdots\\
\vec{0}\\
W_{b+1} (x) \, \vpsi_{b+1}
\end{bmatrix}, 
\end{equation*}
where $\vec{0}$ is the null vector in $\R^l$.

From here we derive that
\begin{equation*}\label{eq2 AL}
 \vpsi_\scale =  
 G_\scale (x ; E) \,
  \begin{bmatrix}
W_a^\top (x) \, \vpsi_{a-1} \\
\vec{0}\\
\vdots\\
\vec{0}
\end{bmatrix} 
+
G_\scale (x ; E) \,
\begin{bmatrix}
\vec{0} \\
\vdots\\
\vec{0}\\
W_{b+1} (x) \, \vpsi_{b+1}
\end{bmatrix} .
\end{equation*}

It follows that for $j \in \scale = [a, b]$ we have the following identity:
\begin{equation}\label{eq star AL}
\vpsi_j = G_{\scale, (j, a) } (x ; E) \, W_a^\top (x) \, \vpsi_{a-1} + G_{\scale, (j, b) } (x ; E) \, W_{b+1} (x) \, \vpsi_{b+1} \, .
\end{equation}

This formula makes it apparent how exponential decay of off-diagonal terms of the Green's function can lead to exponential decay of an extended state. It is more complicated than its counterpart in~\cite{B-G-first} because its terms are $l \times l$ blocks rather than numbers, and there are extra factors related to the weights $W (x)$. However, as the weights are bounded, the relevant estimates are derived in a similar manner, as it can be seen in the two lemmas below.

\medskip

Given two integers $j, k$, we write $j \less k$ when $j \le C k$ for some appropriate absolute constant $C>0$, while 
$j \asymp k$ means that $j \less k$ and $k \less j$.

Let us say that an integer $j$ is {\em well inside} an interval $\scale = [a, b] \subset \N$ if it is far away from its endpoints, for instance if $2 a \le j \le \frac{b}{2}$, so that 
$$\sabs{j-a} = j-a \more  j \quad \text{and} \quad  \sabs{j-b} = b-j \more j .$$

\begin{lemma}\label{AL lemma 1}
Let $N \in \N$ and let $\scale = [a, b] \subset \N$ with $a \asymp N$ and $b \asymp N$. For instance $\scale$ could be $ \left[ \frac{1}{2} N, 2 N \right]$.
Consider $x \in \T$, $\vpsi = (\vpsi_n)_{n\in\Z} \subset \R^l$ and $E\in \R$ such that \eqref{generalized eigenvector} holds.
If $G_\scale (x ; E)$ is a good Green's function and if $j$ is well inside $\scale$ then
$$\snorm{\vpsi_j}_2 \less e^{-c \, j} \, ,$$
where $c = \Bigo (\log \sabs{\la})$.

A similar statement also holds for $\scale \subset \Z_-$.
\end{lemma}

\begin{proof}
Using~\eqref{eq star AL} we have
\begin{align*}
\snorm{\vpsi_j}_2 & \le \norm{ G_{\scale, (j, a) } (x ; E) } \, \norm{ W_a^\top (x) } \, \snorm{ \vpsi_{a-1} }_2 \\
& + \norm{ G_{\scale, (j, b) } (x ; E) } \, \norm{ W_{b+1} (x) } \, \snorm{ \vpsi_{b+1} }_2  .
\end{align*}

Since $G_\scale (x ; E)$ is a good Green's function and $j$ is well inside $\scale$, its $(j, a)$ and $(j, b)$ entries decay exponentially fast in $j$. 

Moreover, since $\vpsi$ is a generalized eigenvector,
$\snorm{\vpsi_{a-1}}_2 \less a^2 \less j^2$,
and similarly for $\vpsi_{b+1}$. 

Recall that for any integer $k$, $W_k (x) = W (x+k\om)$ and $W (x)$ is bounded.

The conclusion then follows.
\end{proof}

\begin{lemma}\label{AL lemma 2}
Let $x \in \T$, $\vpsi = (\vpsi_n)_{n\in\Z} \subset \R^l$ and $E\in \R$ such that \eqref{generalized eigenvector} holds and let $N_0 \in \N$. 
If for some $c>0$
$$\snorm{\vpsi_{N_0}}_2 \le \eta  \quad \text{and} \quad  \snorm{\vpsi_{-N_0}}_2 \le \eta \, ,$$
then
$$\dist \left( E,  \, \text{\rm spectrum } H_{(-N_0, N_0)} (x) \right) \less  \, \eta \, .$$
\end{lemma}

\begin{proof}
Using~\eqref{eq star AL} with $\scale = (-N_0, N_0)$ and $j=0$ we have:
\begin{align*}
1 = \snorm{\vpsi_0}_2 & \le \norm{ G_{\scale, (0, - N_0 + 1) } (x ; E) } \, \norm{W}_\infty \, \snorm{ \vpsi_{-N_0} }_2 \\
& + \norm{ G_{\scale, (0, N_0 - 1) } (x ; E) } \, \norm{W}_\infty \, \snorm{ \vpsi_{N_0} }_2\\
& \less \norm{ G_{\scale} (x ; E) } \, \eta .
\end{align*}

Then
\begin{align*}
\dist\left( E,   \text{\rm spectrum } H_{(-N_0, N_0)} (x) \right) &= \norm{  \left(H_\scale (x) - E \, I\right)^{-1}}^{-1} \\
& = \norm{ G_{\scale} (x ; E) }^{-1} \less \eta ,
\end{align*}
which proves the lemma.
\end{proof}

Fix any $x_0\in \T$ and let $N \in \N$ be any large enough scale. Consider a much larger scale $N' = N^C$. 

The idea of the proof of localization for the operator $H (x_0)$ is to {\em pave}\footnote{Paving $\scale'$ by the intervals $\{ \scale_n \}_n$ means that every point $k \in \scale'$ is {\em well inside} some interval $\scale_n$ in the collection.} an interval $\scale' \supset [ \frac{1}{2} N', 2 N' ]$ of length $\sabs{ \scale' } \asymp N'$ by a collection $\{ \scale_n \}_n$ of intervals
of length 
$\sabs{\scale_n} = N$, where $G_{\scale_n} (x_0 ; E)$ are all good Green's functions. 

An application of the resolvent identity (see Lemma 10.33 in~\cite{Bourgain-book} or  Section 15, Step 3 in~\cite{B-G-first}) implies that $G_{\scale'} (x_0 ; E)$ is also a good Green's function. 

Since $N'$ is well inside the interval $\scale'$, Lemma~\ref{AL lemma 1} implies $\snorm{\vpsi_{N'}}_2 \le e^{- c  \, N'}$, where $c = \Bigo (\log \sabs{\la})$. 

This argument applies to any large enough $N$, and it may be replicated on the negative side $\Z_-$. Hence it establishes exponential decay of the extended state $\vpsi$ and thus Anderson localization.

It remains to explain the paving procedure of a large interval by small intervals whose corresponding  Green's functions are good. It is enough to show that for all $n$ with $\sqrt{N'} \le n \le 2 N'$, there is $0 \le j_n \le \frac{1}{100} N$ such that 
\begin{equation}\label{eq one star} 
G_{[1, N] + j_n} (x_0 + n \om ; E) \text{ is a good Green's function for all } E .
\end{equation}

Indeed, if we denote $\scale_n := [1, N] + j_n + n$, then 
$$G_{\scale_n} (x_0 ; E) = 
G_{[1, N] + j_n} (x_0 + n \om ; E)$$ is a good Green's function. Moreover, 
$$\bigcup \left\{ \scale_n \colon \sqrt{N'} \le n \le 2 N'\right\} \supset \left[ \frac{1}{2} N', 2 N'\right] ,$$
 and consecutive intervals in this collection overlap by a lot (i.e. their intersection has length $\more N$). Hence the paving property holds.

 In order to establish~\eqref{eq one star}, by item~\eqref{Green item 2} of Proposition~\ref{Green's functions estimates} and by Remark~\ref{good intervals}, it is enough to have that 
 \begin{equation}\label{eq two stars}
 x_0 + n \om \notin \bigcup_E \B_N (\om, E) \ \text{ for all } \ \sqrt{N'} \le n \le 2 N'.
 \end{equation}
 
Ensuring that such a stretch of the orbit of $x_0$ under the translation by the frequency $\om$ avoids the corresponding exceptional set above, will require the elimination of a set $\Omega$ of frequencies. The problem is the dependence of these exceptional sets on the eigenvalue $E$. We explain below how to eliminate this dependence on $E$ (for full details, see Section 4 in \cite{B-G-first}).

By Proposition~\ref{Green's functions estimates} item~\eqref{Green item 3}, there are plenty of integers $N^{1+} \le n \le N^p-N$ such that $G_{[1, N] + n} (x_0  ; E) = G_N (x_0 + n \om ; E)$ is a good Green's function. Since there are plenty of integers $N^{1+} \le j \le N^p$ well inside $[1, N] + n$, by Lemma~\ref{AL lemma 1}
$$\snorm{\vpsi_j}_2 \less e^{-c \, j} \, ,$$
where $c = \Bigo (\log \sabs{\la})$.

Repeating this argument on $\Z_-$, we can ensure that there is an integer $N_0$ with $N^{1+} \le N_0 \le N^p$ such that both estimates 
$$\snorm{\vpsi_{N_0}}_2 \le e^{-c \, N_0}  \quad \text{and} \quad  \snorm{\vpsi_{-N_0}}_2 \le e^{-c \, N_0}$$
hold (it would be enough if they held for some $N_0$ and $- N_0'$, where $N_0' \asymp N_0$). 

Applying Lemma~\ref{AL lemma 2} we conclude the following. If $E$ is a generalized eigenvalue for $H (x_0)$, then for every large enough scale $N$, there is $N_0$ with $N^{1+} \le N_0 \le N^p$ such that $E$ is an almost eigenvalue of the finite volume operator  $H_{(-N_0, N_0)} (x_0)$:
$$\dist \left( E,  \, \text{\rm spectrum } H_{(-N_0, N_0)} (x_0) \right) \less  \,  e^{-c \, N_0} <  e^{-c \, N^{1+}} \, .$$ 

We may then replace the condition~\eqref{eq two stars} by 
\begin{equation}\label{eq three stars}
 x_0 + n \om \notin \sbad_N (\om) \ \text{ for all } \ \sqrt{N'} \le n \le 2 N'\,, \text{ where } 
 \end{equation}
$$\sbad_N (\om) := \bigcup \left\{ \B_N (\om, E) \colon E \in \text{\rm spectrum } H_{(-N_0, N_0)} (x_0), \ N_0 \le N^p \right\} .$$
 
This is a union of at most $N^{p+1}$ many sets of measure $< e^{- N^a}$, so
$\meas \left( S_N (\om) \right) < e^{- N^{a'}} $.
Consider the set
$$\sbad_N :=  \left\{ (\om, x) \colon \om \ \text{ Diophantine } \text{ and } \ x \in 
S_N (\om) \right\} .$$

As in the proof of item~\eqref{Green item 3} of Proposition~\ref{Green's functions estimates}, $\sbad_N$ may be regarded as a semi-algebraic set of sub-exponentially  small measure in $N$ but polynomial complexity (we would have to restrict the Diophantine condition~\eqref{DC} to integers $j < N$). Lemma 6.1 in \cite{B-G-first} then provides an estimate on the measure of the set
$$\Omega_N := \left\{ \om \in \T \colon (\om, x_0 + n \om) \in \sbad_N \ \text{ for some } n \sim N'   \right\} ,$$
namely
$$\meas \left( \Omega_N \right) < (N')^{- 1/2} = N^{- C/2} .$$

Since $ \sum_N N^{- C/2} < \infty$, by Borel-Cantelli, the set
$\Omega := \bigcap_k \bigcup_{N\ge k} \Omega_N$ 
has measure zero.
We conclude that if $\om$ is Diophantine and $\om \notin \Omega$, then $\om \notin \Omega_N$ for all large enough $N$. This ensures that~\eqref{eq three stars} holds, and hence it justifies the paving procedure. 

The interested reader may find all the remaining details in~\cite{B-G-first} as well as in Chapter 10 of J. Bourgain's monograph~\cite{Bourgain-book}.

\def\cprime{$'$} \def\cprime{$'$}
\providecommand{\bysame}{\leavevmode\hbox to3em{\hrulefill}\thinspace}
\providecommand{\MR}{\relax\ifhmode\unskip\space\fi MR }
\providecommand{\MRhref}[2]{%
  \href{http://www.ams.org/mathscinet-getitem?mr=#1}{#2}
}
\providecommand{\href}[2]{#2}

\end{document}